\documentclass{elsarticle}

\usepackage{pdfpages}
\usepackage{relsize,xspace}
 \usepackage{xcolor}
 \usepackage{mathtools}
 \usepackage{todonotes}
 \usepackage{comment}
\usepackage{microtype}
\usepackage{amsmath}
\usepackage{amssymb}
\usepackage{amsfonts}
\usepackage{stmaryrd}
\usepackage{bm}
\usepackage{tikz}
\usepackage{refcount}
\usepackage{wrapfig}
\usepackage{thm-restate}
\usepackage{mathrsfs}

\definecolor{blue}{rgb}{0.1,0.2,0.5}
\definecolor{brown}{rgb}{0.6,0.6,0.2}
\usepackage[ocgcolorlinks, linkcolor={blue}, citecolor={brown}]{hyperref}

\usepackage[amsmath,thmmarks,hyperref]{ntheorem}
\usepackage{cleveref}

\crefformat{page}{#2page~#1#3}%
\Crefformat{page}{#2Page~#1#3}%
\crefformat{equation}{#2(#1)#3}%
\Crefformat{equation}{#2(#1)#3}%
\crefformat{figure}{#2Figure~#1#3}%
\Crefformat{figure}{#2Figure~#1#3}%
\crefformat{section}{#2Section~#1#3}
\Crefformat{section}{#2Section~#1#3}
\crefformat{chapter}{#2Chapter~#1#3}
\Crefformat{chapter}{#2Chapter~#1#3}
\crefformat{chapter*}{#2Chapter~#1#3}
\Crefformat{chapter*}{#2Chapter~#1#3}
\crefformat{part}{#2Part~#1#3}
\Crefformat{part}{#2Part~#1#3}
\crefformat{enumi}{#2(#1)#3}
\Crefformat{enumi}{#2(#1)#3}

\usepackage{enumerate}

\usepackage{latexsym}

\theoremnumbering{arabic}
\theoremstyle{plain}
\theoremsymbol{}
\theorembodyfont{\itshape}
\theoremheaderfont{\normalfont\bfseries}
\theoremseparator{.}
\theorempreskip{4pt}
\theorempostskip{3pt}

\newtheorem{theorem}{Theorem}
\crefformat{theorem}{#2Theorem~#1#3}
\Crefformat{theorem}{#2Theorem~#1#3}

\newcommand{\newtheoremwithcrefformat}[2]{%
  \newtheorem{#1}[theorem]{#2}%
  \crefformat{#1}{##2\MakeUppercase#1~##1##3}%
  \Crefformat{#1}{##2\MakeUppercase#1~##1##3}%
}
\newcommand{\newseptheoremwithcrefformat}[2]{%
  \newtheorem{#1}{#2}%
  \crefformat{#1}{##2\MakeUppercase#1~##1##3}%
  \Crefformat{#1}{##2\MakeUppercase#1~##1##3}%
}

\newtheoremwithcrefformat{lemma}{Lemma}
\newtheoremwithcrefformat{proposition}{Proposition}
\newtheoremwithcrefformat{observation}{Observation}
\newtheoremwithcrefformat{conjecture}{Conjecture}
\newtheoremwithcrefformat{corollary}{Corollary}
\newseptheoremwithcrefformat{claim}{Claim}
\theorembodyfont{\upshape}
\newtheoremwithcrefformat{example}{Example}
\newtheoremwithcrefformat{remark}{Remark}
\newseptheoremwithcrefformat{definition}{Definition}

\theoremstyle{nonumberplain}
\theoremheaderfont{\scshape}
\theorembodyfont{\normalfont}
\theoremsymbol{\ensuremath{\square}}
\newtheorem{proof}{Proof}

\theoremsymbol{\ensuremath{\lrcorner}}
\newtheorem{clproof}{Proof}

\def\cqedsymbol{\ifmmode$\lrcorner$\else{\unskip\nobreak\hfil
\penalty50\hskip1em\null\nobreak\hfil$\lrcorner$
\parfillskip=0pt\finalhyphendemerits=0\endgraf}\fi} 

\newcommand{\wcol}{\mathrm{wcol}}

\newcommand{\Oof}{\mathcal{O}}
\newcommand{\Oh}{\mathcal{O}}

\newcommand{\Cc}{\mathscr{C}}
\newcommand{\Dd}{\mathscr{D}}
\newcommand{\Ff}{\mathcal{F}}

\renewcommand{\ker}{\mathrm{ker}}

\newcommand{\cl}{\mathrm{cl}}
\newcommand{\cst}{\alpha}

\newcommand{\fwcol}{f_{\wcol}}
\newcommand{\fker}{f_{\ker}}
\newcommand{\fproj}{f_{\mathrm{proj}}}
\newcommand{\fcl}{f_{\cl}}

\newcommand{\fpaths}{f_{\mathrm{pth}}}

\newcommand{\fuqw}{f_{\mathrm{uqw}}}
\newcommand{\fdual}{f_{\mathrm{dual}}}

\newcommand{\Balls}{\mathrm{Balls}}

\newcommand{\N}{\mathbb{N}}
\newcommand{\R}{\mathbb{R}}

\renewcommand{\phi}{\varphi}
\renewcommand{\epsilon}{\varepsilon}
\newcommand{\eps}{\varepsilon}

\newcommand{\minor}{\preccurlyeq}
\newcommand{\dist}{\mathrm{dist}}

\newcommand{\Exp}{\mathbb{E}}

\newcommand{\projnum}{\mu}

\newcommand{\abs}[1]{\ensuremath{\left\lvert#1\right\rvert}}

\newcommand{\clqed}{\renewcommand{\qedsymbol}{\ensuremath{\lrcorner}}}

\renewcommand{\setminus}{-}
\renewcommand{\leq}{\leqslant}
\renewcommand{\geq}{\geqslant}

\usepackage[absolute]{textpos}

\journal{European Journal of Combinatorics}

\title{Kernelization and approximation of distance-$r$ independent sets on nowhere dense graphs}

\author{Micha\l{} Pilipczuk\tnoteref{t1}}
\address{University of Warsaw, Poland}
\ead{michal.pilipczuk@mimuw.edu.pl}
\author{Sebastian Siebertz\tnoteref{t1}}
\address{University of Bremen}
\ead{siebertz@uni-bremen.de}
\tnotetext[t1]{Our work is supported by the National Science Centre of Poland via POLONEZ grant agreement UMO-2015/19/P/ST6/03998, 
which has received funding from the European Union's Horizon 2020 research and 
innovation programme (Marie Sk\l odowska-Curie grant agreement No.\ 665778).}

\begin{document}


\begin{frontmatter}

\begin{abstract}
  \noindent For a positive integer $r$, a distance-$r$ independent set
  in an undirected graph~$G$ is a set $I\subseteq V(G)$ of vertices
  pairwise at distance greater than $r$, while a distance-$r$
  dominating set is a set $D\subseteq V(G)$ such that every vertex of
  the graph is within distance at most $r$ from a vertex from $D$.
  We study the duality between the maximum size of a
  distance-$2r$ independent set and the minimum size of a distance-$r$
  dominating set in nowhere dense graph classes, as well as the
  kernelization complexity of the distance-$r$ independent set problem
  on these graph classes. Specifically, we prove that the distance-$r$
  independent set problem admits an almost linear kernel on every
  nowhere dense graph~class.
\end{abstract}

\begin{textblock}{5}(11.13, 13.41)
\includegraphics[width=38px]{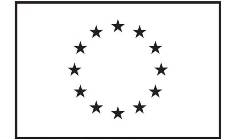}%
\end{textblock}
\end{frontmatter}

\section{Introduction}

\paragraph{Independence and domination} For a graph $G$ and positive
integer $r$, a {\em{distance-$r$ independent set}} in $G$ is a subset
of vertices $I$ whose members are pairwise at distance more than $r$.
On the other hand, a {\em{distance-$r$ dominating set}} in $G$ is a
subset of vertices $D$ such that every vertex of $G$ is at distance at
most $r$ from some member of $D$.  The cases $r=1$ correspond to the
standard notions of an independent and dominating set, respectively.
In this work we will consider combinatorial questions about
distance-$r$ independent and dominating sets, as well as the
computational complexity of the corresponding decision problems
{\sc{Distance-$r$ Independent Set}} and {\sc{\mbox{Distance-$r$} Dominating
    Set}}: given a graph $G$ and integer $k$, decide whether $G$ has a
distance-$r$ independent set of size at least $k$, respectively, a
distance-$r$ dominating set of size at most $k$.

In the following, we denote the minimum size of a
distance-$r$ dominating set in a graph~$G$ by $\gamma_r(G)$ and the
maximum size of a distance-$r$ independent set by
$\alpha_{r}(G)$. Furthermore, if $A\subseteq V(G)$, we write
$\gamma_r(G,A)$ for the minimum size of a {\em{distance-$r$ dominating
    set of $A$}}, i.e., we only require that each vertex of $A$ is at
distance at most $r$ from the dominating set. Similarly, we write
$\alpha_r(G,A)$ for the maximum size of a distance-$r$ independent
subset of $A$.  Observe that for every graph $G$, vertex subset $A$,
and positive integer $r$ we have
$$\alpha_{2r}(G,A)\leq \gamma_r(G,A),$$
because every member of a set that distance-$r$ dominates $A$ can
dominate at most one member of a distance-$2r$ independent subset of
$A$.  The study of a reverse inequality (in the approximate sense) for
certain graph classes is the main combinatorial goal of this work.

Regarding computational complexity, both \textsc{Independent Set} and
\textsc{Dominating Set} are NP-hard~\cite{karp1972reducibility} and
this even holds in very restricted settings, e.g., on planar graphs of
maximum degree~$3$~\cite{garey2002computers,garey1976some}.  Even
worse, under the assumption that $\textsc{P}\neq \textsc{NP}$, for
every $\epsilon>0$, the size of a maximum independent set of an
$n$-vertex graph cannot be approximated in polynomial time within a
factor better than
$\Oof(n^{1-\epsilon})$~\cite{haastad1999clique}. Under the assumption
$\textsc{P}\neq \textsc{NP}$, the domination number of a graph cannot
be approximated in polynomial time within a factor better than
$\Oof(\log n)$~\cite{raz1997sub}. However, it turns out that in
several restricted graph classes the problems can be approximated much
better.  For instance, for fixed $r$ the distance-$r$ variants of both
problems admit a polynomial-time approximation scheme (PTAS) on planar
graphs~\cite{baker1994approximation} and, more generally, in graph
classes with polynomial expansion~\cite{har2017approximation}.  We
will discuss further approximation results later.


\paragraph{Abstract notions of sparsity} In this paper we are going to
study \textsc{\mbox{Distance-$r$} Independent Set} and
\textsc{\mbox{Distance-$r$} Dominating Set} on \emph{nowhere dense}
graph classes.  The notions of {\em{nowhere denseness}} and
{\em{bounded expansion}} are the fundamental definitions of the
sparsity theory introduced by Ne\v{s}et\v{r}il and Ossona de
Mendez~\cite{NesetrilM08,NesetrilM11a}.
Many familiar classes of sparse graphs, like classes of bounded
treewidth, planar graphs, classes of bounded degree, and all classes
that exclude a fixed minor or topological minor have bounded expansion
and are nowhere dense.  In order to facilitate further discussion, we
now recall basic definitions.

Nowhere dense classes and classes of bounded expansion are defined by
imposing restrictions on the graphs that can be found as \emph{bounded
  depth minors} in the class. Formally, for a positive integer $r$, a
graph $H$ with vertex set $\{v_1,\ldots, v_n\}$ is a \emph{depth-$r$
  minor} of a graph~$G$, written $H\minor_r G$, if there are connected
and pairwise vertex disjoint subgraphs $H_1,\ldots, H_n\subseteq G$,
each of radius at most $r$, such that if $v_iv_j\in E(H)$, then there
are $w_i\in V(H_i)$ and $w_j\in V(H_j)$ with $w_iw_j\in E(G)$. Now, a
class $\Cc$ of graphs has bounded expansion if for every positive
integer $r$ and every $H\minor_r G$ for $G\in\Cc$, the edge density
$|E(H)|/|V(H)|$ of $H$ is bounded by some constant
$d(r)$. Furthermore, $\Cc$ is nowhere dense if for every positive
integer $r$ there exists a constant $t(r)$ such that
$K_{t(r)}\not\minor_r G$ for all $G\in\Cc$, where $K_t$ denotes the
complete graph on $t$ vertices.

We call $\Cc$ \emph{effectively nowhere dense}, respectively, of
\emph{effectively bounded expansion}, if the function $t(r)$,
respectively $d(r)$, is computable; such effectiveness is enjoyed by
essentially all natural classes of sparse graphs.  Clearly, every
class of bounded expansion is nowhere dense, but the converse is not
true. For example the class consisting of all graphs $G$ with
$\mathrm{girth}(G)\geq \Delta(G)$ is nowhere dense, however, it does
not have bounded average degree and in particular does not have
bounded expansion, see~\cite{sparsity}.

The duality between independence and domination numbers on classes of
bounded expansion was studied by Dvo\v{r}\'ak~\cite{Dvorak13}, who
proved that for such classes, there is a constant-factor
multiplicative gap between them. More precisely,
Dvo\v{r}\'ak~\cite{Dvorak13} proved that for every class $\Cc$ of
bounded expansion and every positive integer $r$, there exists a
constant $c(r)$ such that every graph $G\in\Cc$ satisfies
\begin{align}\label{eq:1}
\alpha_{2r}(G)\leq \gamma_{r}(G)\leq c(r)\cdot \alpha_{2r}(G).
\end{align}
A by-product of this combinatorial result is a pair of constant-factor
approximation algorithms, for the \textsc{\mbox{Distance-$r$}
  Independent Set} and \textsc{\mbox{Distance-$r$} Dominating Set}
problems on any class of bounded expansion.  One of the goals of this work
is to investigate to what extent the above duality can be lifted to
the more general setting of nowhere dense graph classes.

\paragraph{Fractional parameters} It will be convenient to study the
relation between $\gamma_r(G)$ and $\alpha_{2r}(G)$ through the lenses
of their fractional relaxations.  For a graph $G$ and positive
integer~$r$, consider the following linear programs; here, $N_r(u)$
denotes the set of vertices at distance at most $r$ from $u$ (including $u$ itself).

\begin{align*}
  \gamma_r^\star(G) \coloneqq \min \sum_{v\in V(G)}x_v \quad 
  \text{subject to} \quad  & 
                             \sum_{v\in N_r(u)}x_v \geq 1 
                             \quad \text{for all $u\in V(G)$, and}\\
                           & \hspace{1.2cm} x_v \geq 0 \quad \text{for all $u\in V(G)$}.
\end{align*}
\vspace{-9mm}

\noindent and 
\vspace{2mm}
\begin{align*}
  \alpha_{2r}^\star(G) \coloneqq \max \sum_{v\in V(G)}y_v\quad 
  \text{subject to} \quad & 
                            \sum_{v\in N_r(u)}y_v \leq 1
                            \quad \text{for all $u\in V(G)$, and}\\
                          & \hspace{1.2cm} y_v \geq 0 \quad \text{for all $u\in V(G)$}.
\end{align*}
The two above LPs are dual to each other, and requiring the variables
to be integral yields the values $\gamma_r(G)$ and $\alpha_{2r}(G)$,
respectively. Hence we have
$$\alpha_{2r}(G)\leq \alpha_{2r}^\star(G)=\gamma_r^\star(G)\leq \gamma_r(G).$$

The relationship between $\gamma_r(G)$ and $\alpha_{2r}(G)$ and
their fractional relaxations were discussed by Dvo\v{r}\'ak~\cite{Dvorak19}, and in particular, 
based on the above duality, Dvo\v{r}\'ak~\cite{Dvorak19} improved 
the bounds in Equality~(\ref{eq:1}) for classes of 
bounded expansion and nowhere dense classes.

\paragraph{VC-dimension} Consider a ground set $U$ and a set system
(family) $\mathcal{F}$ consisting of subsets of $U$.  A subset
$X\subseteq U$ is \emph{shattered} by $\mathcal{F}$ if for every
subset $Y\subseteq X$ there exists $F\in \mathcal{F}$ such that
$F\cap X=Y$. The \emph{Vapnik-Chervonenkis dimension}, short
\emph{VC-dimension}, of $\mathcal{F}$ is the maximum size of a set
shattered by $\mathcal{F}$~\cite{chervonenkis1971theory}.  We also
define the notions of a {\em{$2$-shattered}} set and the
{\em{2VC-dimension}} of a set system by restricting subsets
$Y\subseteq X$ considered in the above definition only to subsets of
size exactly $2$. Clearly, the VC-dimension of a set system is upper
bounded by its 2VC-dimension.

A fundamental result about VC-dimension is that in set systems of
bounded VC-dimension the gap between integral and fractional hitting
sets is bounded.  A {\em{hitting set}} of a set system~$\mathcal{F}$
over $U$ is a subset $H\subseteq U$ that intersects every member of
$\mathcal{F}$, while a {\em{fractional hitting set}} is a distribution
of weights from $[0,1]$ among elements of $U$ so that every member of
$\mathcal{F}$ has total weight at least $1$.  Let~$\tau(\mathcal{F})$
and $\tau^\star(\mathcal{F})$ denote the minimum size, respectively
weight, of an integral, respectively fractional, hitting set of
$\mathcal{F}$.

\begin{theorem}[see e.g. \cite{bronnimann1995almost,even2005hitting}]\label{thm:vc-hitting-sets}
  There exists a universal constant $C$ such that for every set system
  $\mathcal{F}$ of VC-dimension at most $d$, we have
$$\tau(\mathcal{F})\leq C\cdot d\cdot \tau^\star(\mathcal{F})\cdot \ln \tau^\star(\mathcal{F}).$$
Moreover, there exists a polynomial-time algorithm that computes a
hitting set of $\mathcal{F}$ of size bounded as above.
\end{theorem}

As proved in~\cite{adler2014interpreting}, any nowhere dense class
$\Cc$ of graphs is {\em{stable}} (a model theoretic property that
describes the complexity of definable set systems); see also~\cite{l2017number} for a
combinatorial proof of this fact.  This, in particular, implies the
following assertion: for every positive integer $r$ there exists a constant
$d(r)$ such that for every $G\in \Cc$ the family of distance-$r$ balls
$$\Balls_r(G)\coloneqq \{\{v\colon \dist_G(u,v)\leq r\}\colon u\in V(G)\},$$
treated as a set system over $V(G)$, has VC-dimension at most $d(r)$.
Combining this with \cref{thm:vc-hitting-sets} shows that
$$\gamma_r(G)\leq C\cdot d(r)\cdot \gamma^\star_r(G)\cdot \ln \gamma^\star_r(G)$$
for every graph $G\in \Cc$.
However, both~\cite{adler2014interpreting} and~\cite{l2017number}
only prove the statement about stability, and consequentely do not provide explicit bounds on the constant~$d(r)$ in the above inequality.

Observe that VC-dimension is a hereditary measure, i.e., for any
subset $A\subseteq U$ of the universe, the VC-dimension of the system
$\Ff\cap A\coloneqq \{F\cap A \colon F\in\Ff\}$ is not larger than the
VC-dimension of $\Ff$. Hence, \cref{thm:vc-hitting-sets}
applied to the set system $\Ff\cap A$ yields
$\tau(\mathcal{F}\cap A)\leq C\cdot d\cdot \tau^\star(\mathcal{F}\cap
A)\cdot \ln \tau^\star(\mathcal{F}\cap A)$,
and we can make the same conclusion about the system stemming from the
$r$-neighborhoods of graphs from a nowhere dense class $\Cc$. That is,
if $A\subseteq V(G)$ for $G\in\Cc$, then
$\gamma_r(G,A)\leq C\cdot d(r)\cdot \gamma^\star_r(G,A)\cdot \ln
\gamma^\star_r(G,A)$,
and the algorithm provided by \cref{thm:vc-hitting-sets} can be
applied to compute a vertex subset that distance-$r$ dominates~$A$
with this size guarantee.

\paragraph{Contribution: duality in nowhere dense classes} We first
study the VC-dimension of systems of radius-$r$ balls in graphs from
nowhere dense graph
classes. 
By following the lines of a recent result of Bousquet and
Thomass\'e~\cite{BousquetT15}, we are able to provide explicit bounds
for the VC-dimension, in fact even for the $2$VC-dimension, of the
$r$th powers of graphs for any nowhere dense class. More precisely, we
prove the following theorem.

\begin{restatable}{theorem}{vc}\label{thm:vc2}
  Let $r\in \N$ and let $G$ be a graph.  If $K_t\not\minor_r G$, then
  the $2$VC-dimension of the set system $\Balls_r(G)$ is at most
  $t-1$.
\end{restatable}

We immediately derive the following; here, $C$ is the constant
provided by \cref{thm:vc-hitting-sets}.

\begin{corollary}\label{cor:lp-relation}
  Let $\Cc$ be a nowhere dense class of graphs such that
  $K_{t(r)}\not\minor_r G$ for all $r\in\N$. Then for every $r\in \N$,
  every $G\in\Cc$ and every $A\subseteq V(G)$ we have
  $$\alpha_{2r}(G,A)\leq \alpha_{2r}^\star(G,A)=\gamma_r^\star(G,A)\leq
  \gamma_r(G,A)\leq C\cdot t(r)\cdot \gamma_r^\star(G,A)\cdot \ln
  \gamma_r^\star(G,A).$$
  Moreover, there exists a polynomial-time algorithm that computes a
  distance-$r$ dominating set of $A$ in $G$ of size bounded as above.
\end{corollary}

\cref{cor:lp-relation} gives an upper bound of
$\Oh(\log \gamma^\star_r(G,A))$ on the multiplicative gap between
$\gamma_r(G,A)$ and $\gamma^\star_r(G,A)$.  For a lower bound, we
prove that one cannot expect that this gap can be bounded by a
constant on every nowhere dense class; recall that this is the case
for classes of bounded expansion~\cite{Dvorak13}.

\begin{restatable}{theorem}{gap}\label{thm:non-const-frac-integral}
  There exists a nowhere dense class $\Cc$ of graphs with the property
  that for every $r\in \N$ we have
  $$\sup_{G\in \Cc}\, \frac{\gamma_r(G)}{\gamma_{r}^\star(G)} = +\infty.$$
\end{restatable}

Finally, we want to investigate the multiplicative gap between
$\gamma_r(G,A)$ and~$\alpha_{2r}(G,A)$.  While the lower bound of
\cref{thm:non-const-frac-integral} asserts that in some nowhere
dense class this gap cannot be bounded by any constant, the upper
bound of \cref{cor:lp-relation} does not provide any upper
bound in terms of $\alpha_{2r}(G,A)$.  To this end, we leverage the
kernelization results for {\sc{Distance-$r$ Dominating Set}} in
nowhere dense classes of~\cite{EickmeyerGKKPRS17} to prove the
following.

\begin{restatable}{theorem}{duality}\label{thm:duality}
  Let $\Cc$ be a nowhere dense class of graphs. There exists a
  function $\fdual\colon \N\times\R\rightarrow \N$ such that for all
  $G\in \Cc$, $A\subseteq V(G)$, $r\in \N$, and $\eps>0$, we have
  $$\gamma_r(G,A)\leq \fdual(r,\eps)\cdot \alpha_{2r+1}(G,A)^{1+\eps}\leq \fdual(r,\eps)\cdot \alpha_{2r}(G,A)^{1+\eps}.$$
  Furthermore, there is a polynomial-time algorithm that given
  $G,A,r,\eps$ as above, computes a distance-$r$ dominating set of $A$
  in $G$ of size bounded as above.
\end{restatable}

Thus, the multiplicative gap is $\Oh(\alpha_{2r}(G,A)^{\eps})$ (and
even $\Oh(\alpha_{2r+1}(G,A)^{\eps})$) for any $\eps>0$.
  
\paragraph{Contribution: kernelization} In the second part of the
paper we turn to the parameterized complexity of \textsc{Distance-$r$
  Independent Set} on nowhere dense classes.  Also from the view of
parameterized complexity both \textsc{Independent Set} and
\textsc{Dominating Set} are hard: parameterized by the target size
$k$, \textsc{Independent Set} is $\textsc{W}[1]$-complete and
\textsc{Dominating Set} is
$\textsc{W}[2]$-complete~\cite{downey1995fixed}. Hence both problems
are not likely to be \emph{fixed-parameter tractable}, i.e., solvable
in time $f(k)\cdot n^c$ on instances of input size $n$, where $f(k)$
is a computable function, depending only on the value of the parameter
$k$ and $c$ is a fixed constant.

Again, it turns out that in several restricted graph classes the
problems become easier to handle.  As far as classes of sparse graphs
are concerned, both \textsc{Distance-$r$ Dominating Set} and
\textsc{Distance-$r$ Independent Set} are expressible in first-order
logic (for fixed $r$), and hence fixed-parameter tractable on any
nowhere dense class of graphs by the meta-theorem of Grohe et
al.~\cite{GroheKS17}.  This was earlier proved in the particular case
of \textsc{Distance-$r$ Dominating Set} by Dawar and
Kreutzer~\cite{DawarK09}.

Once fixed-parameter tractability of a problem on a certain class of
graphs is established, we can ask whether we can go even one step
further by showing the existence of a polynomial (or even linear)
kernel. A \emph{kernelization algorithm}, or a \emph{kernel}, is a
polynomial-time preprocessing algorithm that given an instance $(I,k)$
of a parameterized problem outputs another instance $(I', k')$, which
is equivalent to $(I,k)$, and whose total size $|I'|+k'$ is bounded by
$f(k)$ for some computable function $f$, called the \emph{size} of the
kernel. If $f$ is a polynomial (respectively, linear) function, then
the algorithm is called a \emph{polynomial} (respectively,
\emph{linear}) \emph{kernel}. 
It is known that for decidable problems, the existence
of a kernel is equivalent to fixed-parameter tractability, however, 
in general 
the function $f$ can be arbitrarily~large.

Kernelization of \textsc{Dominating Set} and \textsc{Distance-$r$
  Dominating Set} on sparse graphs classes has received a lot of
attention in the
literature~\cite{alber2004polynomial,bodlaender2016meta,DrangeDFKLPPRVS16,
EibenKMPS19,EickmeyerGKKPRS17,FabianskiPST19,fomin2010bidimensionality,
  fomin2012linear,fomin2018kernels,KreutzerRS17}. In particular,
\textsc{Distance-$r$ Dominating Set} admits a linear kernel on any
class of bounded expansion~\cite{DrangeDFKLPPRVS16} and an almost
linear kernel on any nowhere dense class~\cite{EickmeyerGKKPRS17}.
The kernelization complexity of \textsc{Distance-$r$ Independent Set}
on classes of sparse graphs seems less explored; a linear kernel for
the problem is known on any class excluding a fixed apex
minor~\cite{fomin2010bidimensionality}.

We prove that for every positive integer $r$, \textsc{Distance-$r$
  Independent Set} admits an almost linear kernel on every nowhere
dense class of graphs. In fact, we prove the statement for the slightly
more general, annotated variant of the problem, which finds an
application e.g.\ in the model-checking result of Grohe et
al.~\cite{GroheKS17}.


\begin{restatable}{theorem}{kernel}\label{thm:kernel}
  Let $\Cc$ be a fixed nowhere dense class of graphs, let $r$ be a
  fixed positive integer, and let $\eps>0$ be a fixed real.  Then
  there exists a polynomial-time algorithm with the following
  properties. Given a graph $G\in\Cc$, a vertex subset
  $A\subseteq V(G)$, and a positive integer $k$, the algorithm either
  correctly concludes that $\alpha_{r}(G,A)<k$, or finds a subset
  $Y\subseteq V(G)$ of size at most
  $\fker(r,\epsilon)\cdot k^{1+\epsilon}$, for some function $\fker$
  depending only on $\Cc$, and a subset $B\subseteq Y\cap A$ such
  that 
  $\alpha_{r}(G,A)\geq k \Longleftrightarrow \alpha_{r}(G[Y], B)\geq
  k$.
\end{restatable}

\pagebreak
We remark that in case $\Cc$ is effectively nowhere dense, it is easy
to see that the function $\fker$ above is computable and the algorithm
can be made uniform w.r.t. $r$ and $\eps$: there is one algorithm that
takes $r$ and $\eps$ also on input, instead of a different algorithm
for each choice of $r$ and $\eps$. Furthermore, as
in~\cite{EickmeyerGKKPRS17}, it is easy to follow the lines of the
proof to obtain a linear kernel in case $\Cc$ is a class of bounded
expansion.  That is, the size of the obtained set $Y$ is bounded by
$\Oh(k)$, where the constant hidden in the $\Oh(\cdot)$-notation
depends on $\Cc$ and $r$.

It is not difficult to see that for classes closed under taking
subgraphs, this result cannot be extended further.  More precisely,
similarly as in~\cite{DrangeDFKLPPRVS16} for the case of
{\sc{Distance-$r$ Dominating Set}}, we provide the following lower
bound for completeness.

\begin{restatable}{theorem}{hardness}\label{thm:hardness}
  Let $\Cc$ be a class of graphs that is closed under taking subgraphs
  and that is not nowhere dense. Then there exists an integer $r$
  such that \textsc{Distance-$r$ Independent Set} is
  $\textsc{W}[1]$-hard on $\Cc$.
\end{restatable}



Our proof of \cref{thm:kernel} uses a similar approach
as~\cite{DrangeDFKLPPRVS16, EickmeyerGKKPRS17} for the kernelization
of \textsc{Distance-$r$ Dominating Set}.  We aim to iteratively remove
vertices from~$A$ that are irrelevant for distance-$r$ independent
sets in the following sense. A vertex $v\in A$ is irrelevant if the
following assertion holds: provided $A$ contains a distance-$r$
independent subset of size~$k$, then also $A\setminus \{v\}$ contains
a distance-$r$ independent subset of size $k$.

In order to find such an irrelevant vertex, we start by computing a
good approximation of a distance-$\left\lfloor r/2\right\rfloor$
dominating set~$D$ of $A$. If we do not find a sufficiently small such
set, we can reject the instance, as by \cref{thm:duality} this
implies that there does not exist a large distance-$r$ independent set
in $A$. We now classify the remaining vertices of $A$ with respect to
their interaction with the set $D$ and argue that if $A$ is large we
may find an irrelevant vertex.

We repeat this construction until $A$ becomes small enough (almost
linear in~$k$) and return the resulting set as the set $B$. It now
suffices to add a small set of vertices and edges so that short
distances between the elements of $B$ are exactly preserved. The
result will be the output of the kernelization algorithm.


\paragraph{Organization} We assume familiarity with graph theory and
refer to~\cite{diestel2012graph} for undefined notation.  We provide
basic facts about nowhere dense graph classes in
\cref{sec:prelims} and refer to~\cite{sparsity} for a broader
discussion of the area. We present our results on the VC-dimension of
power graphs in \cref{sec:vc-dim} and the construction of a
nowhere dense class witnessing the non-constant gap between
distance-$r$ domination and distance-$2r$ independence in
\cref{sec:fractionalpacking}.  Finally, we present the
kernelization algorithm for \textsc{Distance-$r$ Independent Set} on
nowhere dense graph classes in \cref{sec:kernel}.
%

\section{Preliminaries}\label{sec:prelims}

We shall need some basic notions and tools for kernelization in
nowhere dense classes used by Eickmeyer et
al.~\cite{EickmeyerGKKPRS17}.  For consistency and completeness of
this paper, we have included these preliminaries also here, and they
are largely taken verbatim from~\cite{EickmeyerGKKPRS17}.

\paragraph{Algorithmic aspects} 
Whenever we say that the running time of some algorithm on a graph $G$
is {\em{polynomial}}, we mean that it is of the form
$\Oof((|V(G)|+|E(G)|)^\cst)$, where $\cst$ is a universal constant
that is independent of $\Cc$, $r$, $\epsilon$, or any other constants
defined in the context. However, the constants hidden in the
$\Oof(\cdot)$-notation may depend on $\Cc$, $r$, and $\epsilon$.

\paragraph{Projections and projection profiles}
Let $G$ be a graph and let $A\subseteq V(G)$ be a subset of
vertices. For vertices $v\in A$ and $u\in V(G)\setminus A$, a path $P$
connecting $u$ and $v$ is called {\em{$A$-avoiding}} if all its
vertices apart from $v$ do not belong to $A$. For a positive integer
$r$, the {\em{$r$-projection}} of any $u\in V(G)\setminus A$ on $A$,
denoted $M^G_r(u,A)$ is the set of all vertices $v\in A$ that can be
connected to $u$ by an $A$-avoiding path of length at most $r$. The
{\em{$r$-projection profile}} of a vertex $u\in V(G)\setminus A$ on
$A$ is a function $\rho^G_r[u,A]$ mapping vertices of $A$ to
$\{1,\ldots,r,\infty\}$, defined as follows: for every $v\in A$, the
value $\rho^G_r[u,A](v)$ is the length of a shortest $A$-avoiding path
connecting $u$ and~$v$, and~$\infty$ in case this length is larger
than $r$. We define
\[\projnum_r(G,A)=|\{\rho_r^G[u,A]\colon u\in V(G)\setminus A\}|\]
to be the number of different $r$-projection profiles realized on $A$. 

One of the main results of~\cite{EickmeyerGKKPRS17} was to show that
in nowhere dense classes the number of realized projection profiles is
small, as stated in the next lemma.

\begin{lemma}[\cite{EickmeyerGKKPRS17}]\label{lem:projection-complexity}
  Let $\Cc$ be a nowhere dense class of graphs. There is a function
  $\fproj\colon \N\times\R\rightarrow\N$ such that for every
  $r\in \N$, $\epsilon>0$, graph $G\in \Cc$, and vertex subset
  $A\subseteq V(G)$, we have
  $\projnum_r(G,A)\leq \fproj(r,\epsilon)\cdot |A|^{1+\epsilon}$.
\end{lemma}

The next lemma states that any vertex subset $X\subseteq V(G)$ can be
``closed'' to a set $X'$ that is not much larger than $X$ such that
all $r$-projections on $X'$ are small.

\begin{lemma}[\cite{DrangeDFKLPPRVS16,EickmeyerGKKPRS17}]\label{lem:closure}
  Let $\Cc$ be a nowhere dense class of graphs. There is a function
  $\fcl\colon \N\times\R\rightarrow \N$ and a polynomial-time
  algorithm that, given $G\in \Cc$, $X\subseteq V(G)$, $r\in \N$, and
  $\epsilon>0$, computes a superset $X'\supseteq X$ of vertices with
  the following properties:
\begin{itemize}
  \item $|X'|\leq \fcl(r,\epsilon)\cdot |X|^{1+\epsilon}$; and
  \item $|M_r^G(u,X')|\leq \fcl(r,\epsilon)\cdot |X|^{\epsilon}$ for
    every $u\in V(G)\setminus X'$.
\end{itemize}
\end{lemma}

The next lemma shows that a set can be closed without increasing its
size too much, so that short distances between its elements are
preserved in the subgraph induced by the closure.

\begin{lemma}[\cite{DrangeDFKLPPRVS16, EickmeyerGKKPRS17}]\label{lem:pathsclosure}
  There is a function $\fpaths\colon \N\times\R\rightarrow \N$ and a
  polynomial-time algorithm that on input $G\in \Cc$,
  $X\subseteq V(G)$, $r\in \N$, and $\epsilon>0$ computes a superset
  $X'\supseteq X$ of vertices with the following properties:
\begin{itemize}
\item whenever $\dist_G(u,v)\leq r$ for $u,v\in X$, then
  $\dist_{G[X']}(u,v)=\dist_G(u,v)$; and
\item $|X'|\leq \fpaths(r,\epsilon)\cdot |X|^{1+\epsilon}$.
\end{itemize}
\end{lemma}

\paragraph{Uniform quasi-wideness}
We will use the characterization of nowhere denseness via the notion
of {\em{uniform quasi-wideness}}, explained next.

\begin{definition}\label{def:uqw}
  A class $\Cc$ is \emph{uniformly quasi-wide} if for every $r\in \N$
  there is a function $N_r\colon \N\times\N\rightarrow \N$ and
  constant $s_r\in \N$ such that for all $r,m\in \N$ and all subsets
  $A\subseteq V(G)$ for $G\in \Cc$ of size $\abs{A}\geq N_r(m)$ there
  is a set $S\subseteq V(G)$ of size $\abs{S}\leq s_r$ and a set
  $B\subseteq A\setminus S$ of size $\abs{B}\geq m$ that is
  $r$-independent in $G-S$.
\end{definition}

It was shown by Ne\v{s}et\v{r}il and Ossona de
Mendez~\cite{NesetrilM11a} that a class~$\Cc$ of graphs is nowhere
dense if and only if it is uniformly quasi-wide.  Recently, Kreutzer
et al.~\cite{KreutzerRS17} proved that in this case the function
$N_r(\cdot)$ can be bounded by a polynomial, with the exponent
depending on $\Cc$ and $r$. More precisely, we will use the following
fact proved later in~\cite{l2017number}; here, the $\Oof_{r,t}(\cdot)$
notation hides factors depending on $r$ and $t$.

\begin{theorem}[\cite{l2017number}]\label{thm:uqw}
  For all $r,t\in \N$ there is a polynomial $N$ with
  $N(m)= \Oof_{r,t}{(m^{{(4t+1)}^{2rt}})}$, such that the following
  holds.  Let $G$ be a graph such that
  $K_t\not\minor_{\lfloor 9r/2\rfloor} G$, and let $A\subseteq V(G)$
  be a vertex subset of size at least $N(m)$, for a given $m$.  Then
  there exists a set $S\subseteq V(G)$ of size $|S|<t$ and a set
  $B\subseteq A\setminus S$ of size $|B|\geq m$ that is
  $r$-independent in $G-S$.  Moreover, given~$G$ and $A$, such sets
  $S$ and $B$ can be computed in time $\Oof_{r,t}(|A|\cdot |E(G)|)$.
\end{theorem}
\section{2VC-Dimension of nowhere dense classes}\label{sec:vc-dim}

In this section we prove \cref{thm:vc2}, which we repeat for
convenience. We remark again that our proof follows the lines of the
work of Bousquet and Thomass{\'{e}}~\cite{BousquetT15}, who proved
that the 2VC-dimension of the set system of all balls (of all radii) in a graph that
excludes $K_t$ as a minor is bounded by $t-1$.  As noted
in~\cite{BousquetT15}, this result was in turn based on the case of
planar graphs considered by Chepoi et al.~\cite{ChepoiEV07}.

\vc*

\begin{proof}
  Assume there is a set $A=\{a_1,\ldots, a_t\}$ of size $t$ such that
  for all subsets $\{i,j\}\subseteq \{1,\ldots,t\}$ of size $2$ there
  is a vertex $v_{ij}$ with
  $\{u\colon \dist(u,v_{ij})\leq r\}\cap A=\{a_i,a_j\}$.  For each
  subset $\{i,j\}\subseteq \{1,\ldots,t\}$ of size $2$, choose a
  vertex $u_{ij}$ so that:
  \begin{enumerate}
  \item \label{p:i}
    $\dist(v_{ij},u_{ij})+\dist(u_{ij},a_i)\leq r$;
  \item \label{p:j}
    $\dist(v_{ij},u_{ij})+\dist(u_{ij},a_j)\leq r$; and
  \item \label{p:min} subject to Conditions \ref{p:i} and
    \ref{p:j}, $\max(\dist(u_{ij},a_i),\dist(u_{ij},a_j))$ is
    minimized.
  \end{enumerate}
  Observe that such $u_{ij}$ must exist since setting $u_{ij}=v_{ij}$
  satisfies the first two conditions.

  Let $P^i_{ij}$ and $P^j_{ij}$ be arbitrarily chosen shortest paths
  between $u_{ij}$ and~$a_i$, and between $u_{ij}$ and~$a_j$,
  respectively.  We now establish some basic properties of paths
  $P^i_{ij}$ and $P^j_{ij}$ following from the choice of $u_{ij}$.

  \begin{claim}\label{cl:ineq}
    For every vertex $x$ on $P^i_{ij}$ we have
    $\dist(v_{ij},x)+\dist(x,a_i)\leq r$, and for every vertex $y$ on
    $P^j_{ij}$ we have $\dist(v_{ij},y)+\dist(y,a_j)\leq r$.
  \end{claim}
  \begin{clproof}
    We prove only the first statement for the second is symmetric.  We
    have
    \begin{align*}
      \dist(v_{ij},x)+\dist(x,a_{i})&\leq
                                      \dist(v_{ij},u_{ij})+\dist(u_{ij},x)+\dist(x,a_{i})\\ 
                                    &=\dist(v_{ij},u_{ij})+\dist(u_{ij},a_{i})\leq r,
    \end{align*}
    where the last equality is due to $x$ lying on a shortest path
    between $u_{ij}$ and $a_i$, and the last inequality is by
    Condition~\ref{p:i}.  \clqed\end{clproof}

  \begin{claim}\label{cl:closer}
    Let $x$ be a vertex on $P^i_{ij}$ that is different from
    $u_{ij}$. Then $\dist(x,a_i)<\dist(x,a_j)$.  Symmetrically, if $y$
    lies on $P^j_{ij}$ and is different from $u_{ij}$, then
    $\dist(y,a_i)>\dist(y,a_j)$.  Consequently, paths $P^i_{ij}$ and
    $P^j_{ij}$ share only one vertex, being the endpoint $u_{ij}$.
  \end{claim}
  \begin{clproof}
    We prove only the first claim, for the second is symmetric and the
    third directly follows from the first two.  Suppose for
    contradiction that $\dist(x,a_i)\geq \dist(x,a_j)$.  By
    \cref{cl:ineq} we have
    $$\dist(v_{ij},x)+\dist(x,a_i)\leq r.$$
    On the other hand, since $\dist(x,a_i)\geq \dist(x,a_j)$, we have
    $$\dist(v_{ij},x)+\dist(x,a_j)\leq\dist(v_{ij},x)+\dist(x,a_i)\leq r.$$
    We conclude that $x$ satisfies Conditions~\ref{p:i} and
    \ref{p:j} from the definition of $u_{ij}$.  However, since
    $x\neq u_{ij}$ and $x$ lies on a shortest path between $u_{ij}$
    and $a_i$, we have $\dist(x,a_i)<\dist(u_{ij},a_i)$.  Therefore,
    $$\dist(x,a_j)\leq \dist(x,a_i)<\dist(u_{ij},a_i)\leq \max(\dist(u_{ij},a_i),\dist(u_{ij},a_j)).$$
    Thus, the existence of $x$ contradicts Condition \ref{p:min}
    from the definition of $u_{ij}$.  \clqed\end{clproof}

  Now define paths $Q^i_{ij}$ and $Q^j_{ij}$ as follows:
  \begin{itemize}
  \item if $\dist(u_{ij},a_i)\leq \dist(u_{ij},a_j)$, then
    $Q^{i}_{ij}=P^{i}_{ij}$ and $Q^{j}_{ij}=P^{j}_{ij} - u_{ij}$;
  \item if $\dist(u_{ij},a_i)>\dist(u_{ij},a_j)$, then
    $Q^{i}_{ij}=P^{i}_{ij} - u_{ij}$ and $Q^{j}_{ij}=P^{j}_{ij}$;
  \end{itemize}
  Thus, by \cref{cl:closer} we have that paths $Q^{i}_{ij}$ and
  $Q^{j}_{ij}$ are disjoint. Moreover, for every vertex $x$ on
  $Q^{i}_{ij}$ we have $\dist(x,a_i)\leq \dist(x,a_j)$, and for every
  vertex $y$ on~$Q^{j}_{ij}$ we have $\dist(y,a_i)\geq \dist(y,a_j)$.

  \begin{claim}\label{cl:intersect}
    Let $\{i,j\}$ and $\{i',j'\}$ be two different subsets of size $2$
    of $\{1,\ldots,t\}$.  Suppose that paths $Q^i_{ij}$ and
    $Q^{i'}_{i'j'}$ intersect.  Then $i=i'$.
  \end{claim}
  \begin{clproof}
    Let $x$ be a vertex lying both on $Q^i_{ij}$ and
    $Q^{i'}_{i'j'}$. We first consider the corner case when
    $x=u_{ij}$.  Suppose first that
    $\dist(v_{ij},x)\geq \dist(v_{i'j'},x)$. Then by
    \cref{cl:ineq} we have
    $$\dist(v_{i'j'},a_i)\leq \dist(v_{i'j'},x)+\dist(x,a_i)\leq \dist(v_{ij},x)+\dist(x,a_i)\leq r,$$
    and analogously $\dist(v_{i'j'},a_{j})\leq r$. However, we assumed
    that $a_{i'}$ and $a_{j'}$ are the only vertices of~$A$ that are
    at distance at most $r$ from $v_{i'j'}$, hence
    $\{i,j\}=\{i',j'\}$, a contradiction. Suppose then that
    $\dist(v_{ij},x)<\dist(v_{i'j'},x)$.  Then we have
    $$\dist(v_{ij},a_{i'})\leq \dist(v_{ij},x)+\dist(x,a_{i'})<\dist(v_{i'j'},x)+\dist(x,a_{i'})\leq r,$$
    where the last equality follows from \cref{cl:ineq}.  Since
    $a_i$ and $a_j$ are the only vertices of $A$ that are at distance
    at most $r$ from $v_{ij}$, we infer that $i'\in \{i,j\}$.  If
    $i'=i$ then we would be done, so suppose $i'=j$.  Since $x=u_{ij}$
    and $x$ lies on $Q^i_{ij}$, by the definition of $Q^i_{ij}$ we
    have that $\dist(x,a_i)\leq
    \dist(x,a_j)=\dist(x,a_{i'})$. Therefore,
    $$\dist(v_{i'j'},a_i)\leq \dist(v_{i'j'},x)+\dist(x,a_{i})\leq \dist(v_{i'j'},x)+\dist(x,a_{i'})\leq r.$$
    where the last inequality follows from \cref{cl:ineq}.
    Again, we assumed that $a_{i'}$ and $a_{j'}$ are the only vertices
    of $A$ that are at distance at most $r$ from $v_{i'j'}$, so
    $i\in \{i',j'\}$. If $i=i'$ then we are done, and otherwise we
    have $i=j'$.  Together with $i'=j$ this implies
    $\{i,j\}=\{i',j'\}$, a contradiction.

    The second corner case when $x=u_{i'j'}$ leads to a contradiction
    in a symmetric manner.

    We now move to the main case when $x\neq u_{ij}$ and
    $x\neq u_{i'j'}$.  Then by \cref{cl:closer} we have
    $\dist(x,a_i)<\dist(x,a_j)$ and $\dist(x,a_{i'})<\dist(x,a_{j'})$.
    By symmetry, without loss of generality assume that
    $\dist(x,a_i)\leq \dist(x,a_{i'})$.  Observe now that
    $$\dist(v_{i'j'},a_i)\leq \dist(v_{i'j'},x)+\dist(x,a_{i})\leq \dist(v_{i'j'},x)+\dist(x,a_{i'})\leq r,$$
    where the last inequality follows from \cref{cl:ineq}.  As we
    assumed that $a_{i'}$ and $a_{j'}$ are the only vertices of $A$
    that are at distance at most $r$ from $v_{i'j'}$, we have
    $i\in \{i',j'\}$.  However, it cannot happen that $i=j'$, because
    $\dist(x,a_{i'})<\dist(x,a_{j'})$ and
    $\dist(x,a_{i'})\geq \dist(x,a_{i})$. We conclude that $i=i'$.
    \clqed\end{clproof}

  For each $i\in \{1,2,\ldots,t\}$ we define $X_i$ to be the union of
  vertex sets of paths $Q^i_{ij}$ for $j\neq i$.  Each of these paths
  has length at most $r$ and has $a_i$ as an endpoint, hence the
  subgraph induced by $X_i$ is connected and has radius at most $r$.
  By \cref{cl:intersect}, sets $X_i$ are pairwise
  disjoint. Finally, observe that for each
  $\{i,j\}\subseteq \{1,\ldots,t\}$ with $i\neq j$, there is an edge
  between a vertex of $Q^{i}_{ij}$ and a vertex of $Q^{j}_{ij}$.  We
  conclude that $(X_i)_{i=1,\ldots,t}$ is a depth-$r$ minor model of
  $K_t$ in~$G$, a contradiction.
\end{proof}
\section{Domination and independence duality in nowhere dense classes}\label{sec:fractionalpacking}

\paragraph{Gap between $\gamma_r$ and $\gamma^\star_r$} 
We first prove \cref{thm:non-const-frac-integral}, which we
repeat for convenience.

\gap*

We will first prove the following auxiliary lemma, which essentially
encompasses the statement for $r=1$.  Here $\Delta(G)$ denotes the
maximum degree of a vertex in $G$, whereas $\mathrm{girth}(G)$ is the
minimum length of a cycle in $G$.

\begin{lemma}\label{lem:Gd-exists}
  For every sufficiently large $d\in\N$ there exists a graph $G_d$,
  say with~$n$ vertices, satisfying the following properties.
  \begin{enumerate}
  \item\label{p:maxdeg} $\Delta(G_d)\leq d$;
  \item\label{p:girth} $\mathrm{girth}(G_d)\geq d$;
  \item\label{p:intg} $\gamma_1(G_d)\geq \frac{\ln d}{2d}\cdot n$; and
  \item\label{p:frac} $\gamma_1^\star(G_d)\leq \frac{2}{d}\cdot n$.
  \end{enumerate}
\end{lemma}
\begin{proof}
  Let $n$ be a large even integer, to be fixed depending on $d$ later.
  We choose the graph $G$ at random using the following random
  procedure.  Consider a set of $dn$ vertices, divided into $n$
  buckets, each containing $d$ vertices.  Choose a matching $M$ on
  those $dn$ vertices uniformly at random.  Collapse each bucket into
  a single vertex, thus creating a multigraph $G_0$ with $dn/2$ edges:
  every edge of $M$ gives rise to one edge in $G_0$ that connects
  the vertices corresponding to the buckets containing the endpoints
  of the original edge. Note that $G_0$ is $d$-regular and may contain
  multiple edges and loops; the latter ones may arise in case some
  edge of $M$ has both endpoints in the same bucket.  This procedure
  of choosing a random $d$-regular multigraph $G_0$ shall be called
  the {\em{bucket model}}.  Finally, we obtain $G$ from $G_0$ as
  follows: for every cycle $C$ of length at most $d$ in $G_0$, pick an
  arbitrary edge of $C$ and remove it.  Note that, in particular, in
  this manner we remove all the loops and reduce the multiplicity of
  every edge to at most $1$; hence $G$ is a simple graph.

  Since $G_0$ is $d$-regular and $G$ is a subgraph of $G_0$, it is
  clear that $\Delta(G)\leq d$; hence Property~\ref{p:maxdeg} is
  always satisfied.  Property~\ref{p:girth} follows directly from the
  construction.  For Properties~\ref{p:intg} and~\ref{p:frac}, we will
  need that for large enough $n$, we actually remove only a sublinear
  (in $n$) number of edges when constructing $G$ from $G_0$.  This
  follows from well-known estimates on the number of short cycles in a
  random regular graph (see
  e.g. Bollob\'as~\cite{bollobas1980probabilistic} and
  Wormald~\cite{wormald1980some, wormald1981asymptotic}), but we give
  a direct proof for completeness.

\begin{claim}
  Let $X$ be a random variable counting the number of different cycles
  of length at most $d$ in $G_0$.  Then for large enough even $n$
  depending on $d$, we have
$$\Exp X \leq \sqrt{n}.$$
\end{claim}
\begin{clproof}
  For an even integer $k$, let $M(k)$ be the number of matchings on
  $k$ vertices, which is
  $$M(k)=(k-1)\cdot (k-3)\cdot\ldots \cdot 3\cdot 1=\frac{k!}{2^{k/2}\cdot (k/2)!}.$$
  By Stirling's approximation, for large enough $k$ we have
  \begin{equation}\label{eq:Mk}
    (2/e)^{k/2}\cdot (k/2)^{k/2}\leq M(k)\leq 2\cdot (2/e)^{k/2}\cdot (k/2)^{k/2}.
  \end{equation}

  For $\ell\in \{1,\ldots,d\}$, let $X_\ell$ be a random variable
  counting the number of cycles of length $\ell$ in $G_0$.  Then
  $X=X_1+\ldots+X_d$. We may estimate the expected value of $X_\ell$
  as follows:
  \begin{equation}\label{eq:Xl}
    \Exp X_\ell \leq n^\ell d^{2\ell}\cdot \frac{M(nd-2\ell)}{M(nd)}.
  \end{equation}
  Indeed, $M(nd)$ is the cardinality of the probabilistic space in the
  bucket model, the factor $n^\ell d^{2\ell}$ is an upper bound on the
  number of ways vertices and edges of a cycle of length $\ell$ may be
  chosen, while $M(nd-2\ell)$ is the number of ways that the other
  edges, outside of this cycle, are chosen.  Combining \eqref{eq:Xl}
  with \eqref{eq:Mk}, for $n$ large enough we have
\begin{align*}
  \Exp X_\ell & \leq 2\cdot n^\ell\cdot d^{2\ell}\cdot (2/e)^{-\ell}\cdot \frac{(nd/2-\ell)^{nd/2-\ell}}{(nd/2)^{nd/2}} \\
              & =    2\cdot n^\ell\cdot d^{2\ell}\cdot (2/e)^{-\ell}\cdot \left(1-\frac{\ell}{nd/2}\right)^{nd/2}\cdot (nd/2-\ell)^{-\ell} \\
              & \leq 2\cdot n^\ell\cdot d^{2\ell}\cdot (2/e)^{-\ell}\cdot e^{-\ell}\cdot (nd/4)^{-\ell} = 2\cdot (2d)^\ell.
\end{align*}
Thus, we have
$$\Exp X = \sum_{\ell=1}^d \Exp X_\ell\leq \sum_{\ell=1}^d 2\cdot (2d)^{\ell}\leq (2d)^{d+1}.$$
Hence, for $n\geq (2d)^{2(d+1)}$ it holds that $\Exp X\leq \sqrt{n}$.
\clqed\end{clproof}

By Markov's inequality, we conclude that with probability at least
$\frac{9}{10}$, $G_0$ contains at most $10\sqrt{n}$ cycles of length
at most $d$, which in particular implies that in this case $G$
contains at most $10\sqrt{n}$ fewer edges than $G_0$.  Removal of a
single edge can increase the domination number and the fractional
domination number of a multigraph by at most $1$; here, when defining
the fractional domination number of a multigraph we determine
domination of a vertex $u$ by verifying whether the sum of weights of
the second endpoints over all edges incident to $u$ is at least
$1$. Hence we have the following.

\begin{claim}\label{cl:close-dom}
With probability at least $\frac{9}{10}$ we have
$$\gamma_1(G_0)\leq \gamma_1(G)\leq \gamma_1(G_0)+10\sqrt{n}\quad\mathrm{and}\quad\gamma^\star_1(G_0)\leq \gamma^\star_1(G)\leq \gamma^\star_1(G_0)+10\sqrt{n}.$$
\end{claim}

\cref{cl:close-dom} essentially reduces checking
Properties~\ref{p:intg} and~\ref{p:frac} for $G$ to checking them for
$G_0$, as the (fractional) domination numbers are almost the same.  On
one hand, we have
$$\gamma^\star_1(G_0)\leq \frac{n}{d},$$
since putting the weight $\frac{1}{d}$ on every vertex of $G_0$ yields a fractional dominating set of size $\frac{n}{d}$. Since $10\sqrt{n}<\frac{n}{d}$ for large enough $n$,
by \cref{cl:close-dom} we conclude that Property~\ref{p:frac} holds with probability at least $\frac{9}{10}$.

To verify that Property~\ref{p:intg} also holds with high probability, we use the results of Alon and Wormald~\cite{alon2010high}, who studied the expected size of a minimum dominating set in a random $d$-regular graph.

\begin{claim}[stated in the proof of Thm 1.2 of Alon and Wormald \cite{alon2010high}]

\mbox{}\\
Fix large enough $d$ and consider even $n$ tending to infinity. 
Then for any $c<1$, the expected number of dominating sets of size at most $\frac{c\ln d}{d}\cdot n$ in a $d$-regular multigraph on $n$ vertices chosen randomly according to the bucket model tends to~$0$.
\end{claim}

In particular, by Markov's inequality, for large enough $n$ the probability that~$G_0$ contains a dominating set of size at most $\frac{\ln d}{4d}\cdot n$ is at most $\frac{1}{10}$.
By \cref{cl:close-dom} we infer that $G$ contains a dominating set of size at most $\frac{\ln d}{4d}\cdot n+10\sqrt{n}$ with probability at most $\frac{2}{10}$,
and this value is upper bounded by $\frac{\ln d}{2d}\cdot n$ for large enough $n$.

We conclude that having fixed $d$ large enough, Properties~\ref{p:intg} and~\ref{p:frac} hold simultaneously for large enough $n$ with probability at least $\frac{7}{10}$. This concludes the proof.
\end{proof}

For a graph $G$ and $r\in \N$, by $G^{(r)}$ we denote the {\em{exact
    $r$-subdivision}} of $G$, which is a graph obtained from $G$ by
subdviding every edge $r-1$ times, i.e., replacing it with a path of
length $r$.  We now lift the statement of \cref{lem:Gd-exists} to
larger radii by considering the following construction; see
\cref{fig:blowup}.

\newcommand{\pendantr}[2]{{#1}^{\langle #2\rangle}}

\begin{definition}\label{def:Gr}
  Let $G$ be a graph and let $r\in\N$.  Construct a graph
  $\pendantr{G}{r}$ from the exact-$r$ subdivision $G^{(r)}$ of $G$ by
  adding two new vertices~$x$ and $y$, connecting~$x$ to every
  subdivision vertex using a path of length $r$, and connecting $y$ to
  $x$ using a path of length $r$.
\end{definition}

\begin{figure}[h]
                \centering
                \def\svgwidth{0.8\columnwidth}
                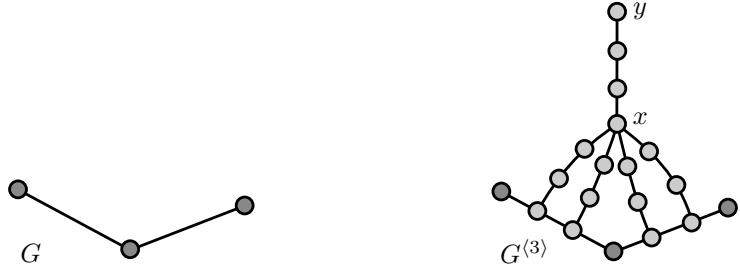
\caption{Construction of $\pendantr{G}{3}$ from $G$, where $G$ is a path on $3$ vertices.}\label{fig:blowup}
\end{figure}

We note the following properties of the above construction. 

\begin{lemma}\label{lem:propGr} The following assertions hold:
  \begin{enumerate}
  \item\label{a:blowup} If $\Cc$ is a nowhere dense graph class, then
    $\{\pendantr{G}{r} \colon G\in\Cc, r\in\N\}$ is nowhere dense as
    well.
  \item\label{a:dom} For every graph $G$ and every $r\in \N$ we have
    $\gamma_r(\pendantr{G}{r})=\gamma_1(G)+1$.
  \item\label{a:ind} For every graph $G$ and every $r\in\N$ we have
    $\alpha_{2r}^\star(\pendantr{G}{r})= \alpha_{2}^\star(G)+1$.
  \end{enumerate}
\end{lemma}
\begin{proof}
  Assertion~\ref{a:blowup} is trivial: if $K_{t}\not\minor_s G$, then
  $K_{t+1}\not\minor_s \pendantr{G}{r}$, for every $r\in \N$.

  For Assertion~\ref{a:dom}, observe first that if $D\subseteq V(G)$
  is a dominating set in $G$, then $D\cup \{x\}$ (as a subset of
  $V(\pendantr{G}{r})$) is a distance-$r$ dominating set of
  $\pendantr{G}{r}$.  For the reverse inequality, let $O=V(G)$ be the
  set of vertices of $\pendantr{G}{r}$ that were originally in $G$,
  and define a function $f\colon V(\pendantr{G}{r})\to O\cup \{x\}$ as
  follows: every vertex of $O$ is mapped to itself, all vertices of
  the paths from $y$ to $x$ are mapped to $x$, and all vertices on any
  length-$r$ path from $x$ (exclusive) to some subdivision vertex~$w$
  (inclusive), where $w$ lies on the length-$r$ path between some
  $u,v\in O$ in~$G^{(r)}$, are mapped to either $u$ or $v$ (chosen
  arbitrarily).  It is not hard to see that if $D$ is a distance-$r$
  dominating set in $\pendantr{G}{r}$, then so is $f(D)$: the
  set~$D$ needs to contain some vertex from the $x$-to-$y$ path, due
  to the necessity of dominating $y$, hence $x\in f(D)$ and $x$
  already $r$-dominates all the vertices of~$\pendantr{G}{r}$ apart
  from~$O$, and~$f$ maps every vertex to a vertex that can only
  dominate more vertices from~$O$.  Since
  $\{x\}\subseteq f(D)\subseteq O\cup \{x\}$ and $x$ does not
  $r$-dominate any vertex of $O$, it follows that
  $f(D)\setminus \{x\}$ is a dominating set in $G$.

Finally, for Assertion~\ref{a:ind}, let $I^\star$ be a fractional
distance-$2$ independent set of~$G$. Then we obtain a fractional
distance-$2r$ independent set of $\pendantr{G}{r}$ of value one larger
by adding weight $1$ to the vertex $y$. Vice versa, let $I^\star$ be a
fractional distance-$2r$ independent set of $\pendantr{G}{r}$. Observe
that all vertices of $V(\pendantr{G}{r})\setminus O$ can collect a
total weight of at most $1$, as they all lie in the $r$-neighborhood
of the vertex $x$. It follows that $I^\star$ restricted to $O$ is a
fractional distance-$2r$ independent set in $G^{(r)}$, which means
that it is also a fractional distance-$2$ independent set in $G$.
\end{proof}


\begin{proof}[Proof of \cref{thm:non-const-frac-integral}]
  The class $\Dd=\{G \colon \mathrm{girth}(G)\geq \Delta(G)\}$ is
  known to be nowhere dense~\cite{sparsity}.  By
  \cref{lem:Gd-exists}, for every large enough $d\in \N$ there is
  a graph $G_d\in \Dd$ with
  $\frac{\gamma_1(Gd)}{\gamma^\star(G_d)}\geq \frac{\ln d}{4}$.  It
  follows that
\[\sup_{G\in \Dd}\, \frac{\gamma_1(G)}{\gamma_{1}^\star(G)} = +\infty.\]
By \cref{lem:propGr}, Assertion~\ref{a:blowup}, the class
$\Cc\coloneqq \{\pendantr{G}{r} \colon G\in\Dd, r\in\N\}$ is nowhere
dense. By \cref{lem:propGr}, Assertions~\ref{a:dom}
and~\ref{a:ind}, and the fact that
$\gamma_1^\star(G)=\alpha_2^\star(G)$ for all graphs $G$, for every
$r\in\N$ we have
\[\sup_{\pendantr{G}{r}\in \Cc}\, \frac{\gamma_r(\pendantr{G}{r})}{\gamma_{r}^\star(\pendantr{G}{r})}=
\sup_{G\in \Dd}\, \frac{\gamma_1(G)+1}{\gamma_{1}^\star(G)+1}\geq
\sup_{G\in \Dd}\, \frac{\gamma_1(G)}{2\gamma_{1}^\star(G)}=
\frac{1}{2}\sup_{G\in \Dd}\,
\frac{\gamma_1(G)}{\gamma_{1}^\star(G)}=+\infty.\]
This concludes the proof.
\end{proof}

\paragraph{Gap between $\alpha_{2r}$ and $\gamma_r$}
We now prove \cref{thm:duality}, which we repeat for
convenience.

\duality* 

We are going to make use of the following kernelization result of
\cite{EickmeyerGKKPRS17} for distance-$r$ dominating sets.

\begin{lemma}[\cite{EickmeyerGKKPRS17}]\label{lem:kerds-orig}
  Let $\Cc$ be a nowhere dense class of graphs. There exists a
  function $\fker\colon \N\times \R\rightarrow\N$ and a
  polynomial-time algorithm with the following properties. Given a
  graph $G\in\Cc$, a vertex subset $A\subseteq V(G)$, positive
  integers $k,r$ and a real number $\epsilon>0$, the algorithm either
  correctly concludes that $\gamma_{r}(G,A)>k$, or finds a subset
  $Y\subseteq V(G)$ of size at most
  $\fker(r,\epsilon)\cdot k^{1+\epsilon}$ and a subset
  $B\subseteq Y\cap A$ such that
  $$\min(\gamma_{r}(G,A),k) = \min(\gamma_{r}(G[Y], B),k).$$
  Moreover, every minimum size subset of $Y$ that distance-$r$ dominates $B$ in $G[Y]$ also distance-$r$ dominates $A$ in $G$.
\end{lemma}

We remark that the statement of this result given in
\cite{EickmeyerGKKPRS17} is somewhat weaker, but the above formulation
follows readily from the proof.  First, in~\cite{EickmeyerGKKPRS17}
the result is stated for distance-$r$ dominating sets of the whole
graph $G$, i.e., for $A=V(G)$; however, throughout the whole
reasoning, the more general problem of dominating a subset of vertices
is considered and it is straightforward to lift the argument to this
setting.  Second, in~\cite{EickmeyerGKKPRS17}, only assertion
$\gamma_r(G,A)\leq k\Leftrightarrow \gamma_r(G[Y],B)\leq k$ is stated,
whereas the stronger assertion that the distance-$r$ domination
numbers are equal in case they are both bounded by $k$ follows
immediately from the proof.  Finally, the fact that every minimum size subset of~$Y$ that distance-$r$ dominates $B$ in $G[Y]$ also distance-$r$
dominates $A$ in $G$ is the key property used to prove the above.

Let us remark one more property of the set $Y$ provided by
\cref{lem:kerds-orig} that follows implicitly from the proof in
\cite{EickmeyerGKKPRS17}: adding any vertices of $G$ to $Y$ does not
change the asserted properties of $Y$.  We may hence apply
\cref{lem:pathsclosure} with parameter $2r+1$ to the set $B$ in
the graph $G$, thus obtaining its superset $B'$, and add all the
vertices of $B'$ to the set $Y$. Thus, we additionally ensure that
distances up to value $2r+1$ between elements of $B$ are preserved in
$G[Y]$.  We derive the following lemma.

\begin{lemma}\label{lem:kerds}
  Let $\Cc$ be a nowhere dense class of graphs. There exists a
  function $\fker\colon \N\times \R\rightarrow\N$ and a
  polynomial-time algorithm with the following properties. Given a
  graph $G\in\Cc$, a vertex subset $A\subseteq V(G)$, positive
  integers $k,r$ and a real number $\epsilon>0$, the algorithm either
  correctly concludes that $\gamma_{r}(G,A)>k$, or finds a subset
  $Y\subseteq V(G)$ of size at most
  $\fker(r,\epsilon)\cdot k^{1+\epsilon}$ and a subset
  $B\subseteq Y\cap A$ such that
  $\alpha_{2r+1}(G,B)=\alpha_{2r+1}(G[Y], B)$ and
  $$\min(\gamma_{r}(G,A),k) = \min(\gamma_{r}(G[Y], B),k).$$
  Moreover, every minimum size subset of $Y$ that $r$-dominates $B$ in $G[Y]$ also $r$-dominates $A$ in~$G$.
\end{lemma}

\pagebreak
We now apply Dvo\v{r}\'ak's duality theorem~\cite{Dvorak13} for
bounded expansion classes, which is in fact stated in terms of a
parameter called the \emph{weak-$r$ coloring number}, denoted
$\wcol_r(\cdot)$. 
We refer to~\cite{Dvorak13} for a formal definition
of this parameter and only note the following lemma, which
characterizes nowhere dense classes in terms of weak coloring numbers.

\begin{lemma}[Zhu~\cite{zhu2009colouring}]\label{lem:zhu-wcol}
  A class of graphs $\Cc$ is nowhere dense if and only if there exists
  a function $\fwcol\colon \N\times\R\rightarrow \N$ such that for all
  $r\in\N$ and all $\epsilon>0$ and all $n$-vertex graphs
  $H\subseteq G$ for $G\in\Cc$ we have
  $\wcol_r(H)\leq \fwcol(r,\epsilon)\cdot n^\epsilon$. 
\end{lemma}

The key ingredient of our proof is the following result of
Dvo\v{r}\'ak~\cite{Dvorak13}, which relates domination and
independence numbers in graphs with bounded weak coloring numbers.

\begin{lemma}[Dvo\v{r}\'ak~\cite{Dvorak13}, see also~\cite{DrangeDFKLPPRVS16}]\label{lem:dv-wcol}
For every graph $G$ and vertex subset $A\subseteq V(G)$,
\[\gamma_r(G,A)\leq \wcol_{2r+1}(G)^2 \cdot \alpha_{2r+1}(G,A).\]
\end{lemma}

We are ready to give the proof of \cref{thm:duality}.

\begin{proof}[Proof of \cref{thm:duality}]
  Given $G\in\Cc$ and $A\subseteq V(G)$, let
  $\gamma\coloneqq \gamma_r(G,A)$. According to 
\cref{cor:lp-relation} we can approximate a distance-$r$ dominating
set of $A$ of size bounded by $k\coloneqq C\cdot t(r)\cdot \gamma_r^\star(G,A)\cdot \ln
  \gamma_r^\star(G,A)$ in polynomial time, where $C$ is the 
  constant and $t$ is 
  the function from the corollary.

  We define $\delta>0$ as a constant
  depending on $\epsilon$, which will be determined at the end of the
  proof. Apply \cref{lem:kerds} to find a subset
  $Y\subseteq V(G)$ of size at most
  $\fker(r,\delta)\cdot k^{1+\delta}$ and a subset
  $B\subseteq Y\cap A$ such that $\gamma_{r}(G[Y], B)=\gamma$ and
  $\alpha_{2r+1}(G,B)=\alpha_{2r+1}(G[Y], B)$. As $G[Y]$ is a subgraph
  of $G$, according to \cref{lem:zhu-wcol} we have
 \begin{align*}
   \wcol_{2r+1}& (G[Y])\leq \\ &\fwcol(2r+1,\delta)\cdot (\fker(r,\delta)\cdot k^{1+\delta})^\delta\leq \fwcol(2r+1,\delta)\cdot \fker(r,\delta)^\delta\cdot k^{2\delta}.
 \end{align*} By \cref{lem:dv-wcol}, we have 
 \begin{align*}
   \gamma = \gamma_r(G,A) & = \gamma_r(G[Y],B)\leq \wcol_{2r+1}(G[Y])^2\cdot \alpha_{2r+1}(G[Y],B)\\
                    & = \wcol_{2r+1}(G[Y])^2\cdot \alpha_{2r+1}(G,B)\leq \wcol_{2r+1}(G[Y])^2\cdot 
                      \alpha_{2r+1}(G,A)\\
                    & \leq (\fwcol(2r+1,\delta)\cdot\fker(r,\delta)^\delta\cdot k^{2\delta})^2\cdot 
                      \alpha_{2r+1}(G,A).
 \end{align*}
Hence
\begin{align*}
  \gamma^{1-4\delta} & \leq \fwcol(2r+1,\delta)^2\cdot\fker(r,\delta)^{2\delta}\cdot 
                  \alpha_{2r+1}(G,A),
\end{align*}
which means that
\begin{align*}
  \gamma & \leq \big(\fwcol(2r+1,\delta)^2\cdot\fker(r,\delta)^{2\delta}\cdot 
      \alpha_{2r+1}(G,A)\big)^{1/(1-4\delta)}\\
    & \leq \big(\fwcol(2r+1,\delta)^2\cdot\fker(r,\delta)^{2\delta}\cdot 
      \alpha_{2r+1}(G,A)\big)^{1+8\delta}.
\end{align*}

\pagebreak
Hence the computed distance-$r$ dominating set has size 
$k\leq C\cdot t(r)\cdot \gamma_r^\star(G,A)\cdot \ln
  \gamma_r^\star(G,A)\leq C\cdot t(r)\cdot \gamma \cdot \ln
  \gamma\leq C\cdot t(r)\cdot \frac{1}{\delta}\cdot \gamma^{1+\delta}$, as $\ln \gamma\leq \frac{1}{\delta}\cdot\gamma^{\delta}$. For 
  $0<\delta<1$ we have $(1+8\delta)(1+\delta)\leq 1+17\delta$ 
  and we can conclude by setting $\delta=\epsilon/17$ and
$\fdual(r,\epsilon)=\big(C\cdot t(r)\cdot \frac{1}{\delta}\cdot \fwcol(2r+1,\delta)^2\cdot
\fker(r,\delta)^{2\delta}\big)^{1+17\delta}$.
%
\end{proof}

We want to stress again that the sets $B$ and $Y$ 
in the construction of \cref{lem:kerds} when applied with 
$k\geq \gamma_{r}(G,A)$ satisfy $\gamma_{r}(G,A) = 
\gamma_{r}(G[Y], B)$. Moreover, every minimum size 
subset of $Y$ that $r$-dominates $B$ in $G[Y]$ also 
$r$-dominates $A$ in~$G$. In the proof of \cref{thm:duality}
we used the algorithm of \cref{cor:lp-relation} to find a 
distance-$r$ dominating set of $A$ in $G$, as approximations
to distance-$r$ dominating sets of $B$ in $G[Y]$ can 
in general not be lifted to distance-$r$ dominating sets 
of $A$ in $G$.
\section{Kernelization}\label{sec:kernel}

We now come to the proof of \cref{thm:kernel}, which we repeat
for convenience.

\kernel*

For the rest of this section let us fix a nowhere dense class $\Cc$,
$k,r\in \N$ and $\epsilon>0$, as well as a graph $G\in\Cc$ and
$A\subseteq V(G)$. We will carry out the construction with a real
$\delta>0$ that will be defined as a function of $\epsilon$ at the end
of the proof. Let $d\coloneqq \left\lfloor r/2\right\rfloor$.

\smallskip Using the algorithm of \cref{thm:duality}, we first
compute an approximate \mbox{distance-$d$} dominating set $D$ for $A$.
We have $|D|\leq \fdual(d,\delta)\cdot \alpha_r(G,A)^{1+\delta}$,
hence if $|D|>\fdual(d,\delta)\cdot k^{1+\delta}$, then we may
immediately conclude that $\alpha_r(G,A)>k$ and terminate the
algorithm.  Therefore, from now on we may assume that
$|D|\leq \fdual(d,\delta)\cdot k^{1+\delta}$.

\smallskip Now, using the algorithm of \cref{lem:closure} (with
parameters $2r$ and $\delta$) we compute a vertex subset $Z$ with the
following properties:
\begin{itemize}
\item $D\subseteq Z$, in particular, $Z$ is a distance-$d$ dominating
  set of $A$;
\item $|Z|\leq \fcl(2r,\delta)\cdot |D|^{1+\delta}$; and
\item $|M_{2r}^G(u,Z)|\leq \fcl(2r,\delta)\cdot |D|^{\delta}$ for each
  $u\in V(G)\setminus Z$.
\end{itemize}
Here, $\fcl$ is the function provided by \cref{lem:closure} for
the class $\Cc$.

We now classify the elements of $A\setminus Z$ with respect to their
$2r$-projections profiles onto $Z$. More precisely, we define the
following equivalence relation:
\[u\sim_Z v \quad \Longleftrightarrow \quad
\rho_{2r}^G[u,Z]=\rho_{2r}^G[v,Z].\]

According to \cref{lem:projection-complexity} (applied with
parameters $2r$ and $\delta$), there exists a function~$\fproj$ such
that this equivalence relation has at most
$\fproj(2r,\delta)\cdot |Z|^{1+\delta}$ equivalence classes. Observe
that the empty projection profile (i.e., one that maps all of $Z$ to
$\infty$) is not realized, as $Z$ is a distance-$d$ dominating
set of~$A$.

Examine the sizes of equivalence classes of $\sim_Z$ and suppose that some class~$K$ of $\sim_Z$ has more than 
$$N(4r,\fcl(2r,\delta)\cdot |D|^{\delta}+(r+1)^{s(4r)}\cdot (s(4r)+1)+1)$$ elements, 
where $N$ and $s$ are the functions from \cref{def:uqw},
characterizing $\Cc$ as uniformly quasi-wide; note here that for fixed
$r$, $s(4r)$ is a fixed constant and $N(4r,\cdot)$ is a fixed
polynomial, both hard-codable in the algorithm, so the above number
can be computed.  Then we proceed as follows. In the following, we
simply write $s$ for $s(4r)$.  We apply the algorithm of
\cref{thm:uqw} to find a subset $S\subseteq V(G)$ of size at
most $s$ and a set $L\subseteq K\setminus S$ of size at least
$\fcl(2r,\delta)\cdot |D|^{\delta}+(r+1)^{s}\cdot (s+1)+1$ that is
$4r$-independent in $G-S$.

\begin{claim}
  There are at most $\fcl(2r,\delta)\cdot |D|^{\delta}$ elements in
  $L$ that are at distance at most $2r$ from $Z$ in the graph $G-S$.
\end{claim}
\begin{clproof}
  No two elements of $L$ can be connected by a path of length at
  most~$2r$ in $G-S$ to the same element of $Z$, as by assumption, the
  elements of $L$ are $4r$-independent in $G-S$.  However, every
  element of $Z$ that is at distance at most~$2r$ in $G-S$ from some
  element of $L$ must belong to the common $2r$-projection of vertices
  of $L$ onto $Z$.  This projection has size at most
  $\fcl(2r,\delta)\cdot |D|^{\delta}$, so there can be at most this
  many vertices in $L$ that are at distance at most $2r$ from $Z$ in
  $G-S$.  \clqed\end{clproof}

Hence, at least $(r+1)^{s}\cdot (s+1)+1$ elements of $L$ cannot be
connected to~$Z$ by a path of length at most~$2r$ in $G-S$.  We
classify these elements with respect to their $r$-projections onto
$S$, that is, we define the following equivalence relation:
\[u\sim_S v \quad \Longleftrightarrow \quad \rho_{r}^G[u,S]=\rho_{r}^G[v,S].\]

As $S$ has size at most $s$ and $\rho_r^G[u,S]$ is a function mapping
from $S$ to $\{1,\ldots, r,\infty\}$, this  relation has at
most $(r+1)^{s}$ equivalence classes. Hence, there exists at least one
equivalence class $L'$ of $\sim_S$ that has at least $s+2$ elements.

\begin{claim}
  Every element of $L'$ is irrelevant, i.e., for every $a\in L'$ there
  exists a distance-$r$ independent subset of $A$ of size $k$ if and
  only if there exists a distance-$r$ independent subset of
  $A\setminus\{a\}$ of size $k$
\end{claim}
\begin{clproof}
  Let $I\subseteq A$ be a distance-$r$ independent set that contains
  an element $a\in L'$. We show that we can replace $a$ with some
  other element $a'\in L'$ to obtain a distance-$r$ independent set
  $I'$ of the same size.

  For this, we show that if we cannot find such an element $a'\in L'$,
  then $I$ was not a distance-$r$ independent set.  So assume that the
  $r$-neighborhood (in $G$) of every $a'\in L'\setminus\{a\}$ contains
  an element $b\in I\setminus\{a\}$ that prevents choosing~$a'$ into the
  distance-$r$ independent set.  For each $a'\in L'$ fix such an
  element $b(a')$ of $I\setminus\{a\}$.

  First assume that some path of length at most $r$ between $a'$ and
  $b(a')$ contains an element of $S$ and let $t$ be the one closest to
  $a'$. As $a\sim_S a'$, the vertex $a$ is at the same distance from
  $t$ as $a'$, which implies that $a$ is at distance at most~$r$
  from~$b(a')$. However, by assumption $a\in I$ and 
  $b(a')\in I\setminus\{a\}$,
  which implies that~$I$ is not a distance-$r$ independent set. Hence
  in the following assume that no path of length at most $r$ between
  $a'$ and $b(a')$ contains an element of $S$, for any
  $a'\in L'\setminus \{a\}$. Note that this implies $b(a')\neq b(a'')$
  for $a'\neq a''$, as the elements of $L'$ are $4r$-independent in 
  $G-S$. 

  We show that this assumption implies that every element
  $b(a')\in I\setminus\{a\}$, for $a'\in L'\setminus \{a\}$, can be connected in $G$
  by a path of length at most $d$ to a vertex of~$S$. To see this,
  recall that every element of $L'$ can be connected in $G$ by a path
  of length at most $d$ to some $z\in Z$, as $Z$ is a distance-$d$
  dominating set of~$A$ in~$G$.  As $a'$ and $b(a')$ are at distance
  at most $r$, the distance between~$a'$ and $z$ is at most
  $r+d\leq 2r$.  However, the elements of $L'$ are at distance more
  than~$2r$ from $Z$ in $G-S$. Hence, as we assumed that no path of
  length at most~$r$ between $a'$ and $b(a')$ contains an element of
  $S$, $b(a')$ must be within distance at most $d$ to some $t\in S$.

  Now, since $|S|\leq s$ and $|L'|\geq s+2$, at least $2$ elements
  $b(a')$ and $b(a'')$ in $I\setminus \{a\}$ associated with $a'$ and $a''$ in
  $L'-\{a\}$ (we noted above that $b(a')\neq b(a'')$ for $a'\neq a''$) 
  must be at distance at most $d$ from the same element of
  $S$. Hence, their distance in $G$ is at most $2d\leq r$, which
  implies that $I$ is not a distance-$r$ independent set, a
  contradiction. This finishes the proof of the claim.
  \clqed\end{clproof}

Hence, if $A$ is large, we can safely remove any element $a'\in L'$
from $A$.  We repeat the whole procedure as long as this is possible,
i.e.\ until we do not find an equivalence class $K$ of size
$N(\fcl(2r,\delta)\cdot |D|^{\delta}+(r+1)^{4s}\cdot (s+1)+1)$
anymore.  Then, we return the current set $A$ as the set $B$ whose
existence we claimed in the theorem. Obviously, the classification
procedure can be carried out in polynomial time and has to be repeated
at most $|A|$ times.  Hence, we can compute the set~$B$ in polynomial
time.

According to \cref{lem:pathsclosure} there exists a function
$\fpaths$ and a polynomial-time algorithm that computes for the set
$B$ a superset $Y$ such that
\begin{itemize}
\item whenever $\dist_G(u,v)\leq r$ for $u,v\in B$, then
  $\dist_{G[Y]}(u,v)=\dist_G(u,v)$; and
\item $|Y|\leq \fpaths(r,\epsilon)\cdot |B|^{1+\epsilon}$.
\end{itemize}

We compute such a set $Y$. The following observation is immediate.

\begin{claim}
  The sets $A,B$ and $Y$ satisfy the following property.  There exists
  an $r$-independent subset $I\subseteq A$ of size $k$ in $G$ if and
  only if there exists an $r$-independent subset $I'\subseteq B$ of
  size $k$ in $G[Y]$.
\end{claim}

Hence the set $B$ and $Y$ provide us with a kernel, as desired.  It
remains to choose $\delta$ so that we obtain the claimed bounds on the
size of the kernel. We have
\begin{align*}
  |Y| & \leq \fpaths(r,\epsilon)\cdot |B|^{1+\delta}\\
      & \leq \fpaths(r,\epsilon)\cdot \big(\fproj(2r,\delta)\cdot |Z|^{1+\delta}\cdot N(\fcl(2r,\delta)\cdot |D|^{\delta}+(r+1)^{s(4r)}\\
      &\hspace{8.27cm} \cdot(s(4r)+1)+1) \big) ^{1+\delta}.
\end{align*}
Let $p\coloneqq (4t(r)+1)^{2r\cdot t(r)}$, where graphs from $\Cc$
exclude $K_{t(r)}$ as a depth-$r$ minor. Using the bounds of
\cref{thm:uqw}, we can now define a function $\fuqw$ so that
$$N(\fcl(2r,\delta)\cdot |D|^{\delta}+(r+1)^{s(4r)}\cdot (s(4r)+1)+1)\leq \fuqw(r,\delta)\cdot |D|^{\delta\cdot p}.$$ Hence
\begin{align*}
  |Y| & \leq \fpaths(r,\epsilon)\cdot \big(\fproj(2r,\delta)\cdot |Z|^{1+\delta}\cdot \fuqw(r,\delta)\cdot |D|^{\delta\cdot p} \big) ^{1+\delta}.
\end{align*}
Now using $|Z|\leq \fcl(2r,\delta)\cdot |D|^{1+\delta}$, $|D|\leq \fdual(r,\delta)\cdot k^{1+\delta}$ and $\delta^2\leq \delta$, we can define $\fker$ so that 
\[|Y|\leq \fker(r,\delta)\cdot k^{1+23p\delta}.\]
By defining $\delta\coloneqq \epsilon/23p$ we conclude the proof of
\cref{thm:kernel}.

\section{Hardness on somewhere dense classes}

We finally prove \cref{thm:hardness}, which we repeat for
convenience.

\hardness*

We shall use the following well-known characterization of somewhere
dense graph classes; recall that $G^{(r)}$ denotes the exact
$r$-subdivision of a graph $G$.

\begin{lemma}[\cite{NesetrilM11a}]\label{lem:sd-all}
  Let $\Cc$ be somewhere dense graph class that is closed under taking
  subgraphs.  Then there exists $r\in\N$ such that $G^{(r)}\in \Cc$
  for all graphs $G$.
\end{lemma}

The hardness proof is based on a very simply reduction.  Using
\cref{lem:sd-all} we will give an easy reduction from the
classical {\sc{Independent Set}} problem, which is known to be
$\textsc{W}[1]$-hard.

\begin{lemma}[\cite{downey1995fixed}]
  \textsc{Independent Set} is $\textsc{W}[1]$-hard on the class of all
  graphs.
\end{lemma}


\setcounter{claim}{0}

\begin{figure}[h]
                \centering
                \def\svgwidth{0.8\columnwidth}
                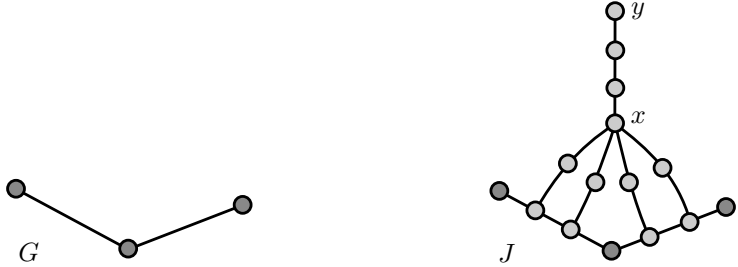
\caption{Construction of $J$ from $G$, where $G$ is a path on $3$ vertices. The graph $H$ is the exact $r$-subdivision of $J$.}\label{fig:reduction}
\end{figure}

\begin{proof}[Proof of \cref{thm:hardness}]
  Let $r\in\N$ be the integer given by \cref{lem:sd-all} for the
  class $\Cc$, that is, $G^{(r)}\in \Cc$ for every graph $G$. Now fix
  an arbitrary graph~$G$.  We construct in polynomial time the
  following graph $H$; see \cref{fig:reduction}. We first
  construct $G^{(3)}$, i.e., we replace in $G$ every edge by a path of
  length $3$. Next, we add two new vertices $x$ and $y$. We connect
  $x$ to each of the previously added subdivision vertices using a
  path of length $2$, and we connect $y$ to $x$ using a path of
  length~$3$. Denote the resulting graph by $J$. Now construct
  $H\coloneqq J^{(r)}$, which by assumption on~$r$ belongs to~$\Cc$.
  We denote the vertices of $H$ that are also vertices of $G$ by $O$
  (for {\em{original}}).

  \pagebreak The following claims summarize the main distance
  properties of $H$.

\begin{claim}\label{obs:distO}
  For all $u,v\in O$ we have
  $\dist_G(u,v)\geq 2\Leftrightarrow\dist_H(u,v)= 6r$.
\end{claim}
\begin{clproof}
  It is straightforward to see that the distance between any two
  vertices of $O$ is exactly $3r$ if they are adjacent in 
  $G$ and exactly $6r$ (possibly via a path using $x$) if they have 
  distance at least $2$ in $G$. The
  claim follows.  \clqed\end{clproof}

\begin{claim}\label{obs:distR}
  For all $u,v\in V(H)\setminus O$ we have $\dist_H(u,v)<6r$.
\end{claim}
\begin{clproof}
  The distance from any $u\in V(H)\setminus (O\cup\{y\})$ to $x$ is at
  most $(r-1)+2r=3r-1$, while the distance from $y$ to $x$ is
  $3r$. Hence, the distance between any $u,v\in V(H)\setminus O$ is
  smaller than $6r$.  \clqed\end{clproof}

We can hence relate the size of a distance-$(6r-1)$ independent set in
$H$ to the size of an independent set in $G$ as follows.

\begin{claim}
We have $\alpha_{6r-1}(H)=\alpha_1(G)+1$. 
\end{claim}
\begin{clproof}
  Let $I$ be a distance-$(6r-1)$ independent set in $H$. By
  \cref{obs:distR}, $I$ can contain at most one vertex of
  $V(H)\setminus O$. We may assume that this vertex is the vertex $y$,
  because $y$ is at distance $6r$ from every vertex of $O$.
  Now $I'\coloneqq I\setminus \{y\}\subseteq O$ is an independent set
  in $G$ by \cref{obs:distO}.

  Conversely, let $I$ be an independent set in $G$. Then $I\cup \{y\}$
  is a distance-$(6r-1)$ independent set in $H$.  \clqed\end{clproof}

Hence, finding an independent set of size $k$ in an arbitrary graph
$G$ reduces to finding a distance-$(6r-1)$ independent set of size
$k+1$ on a polynomial-time computable graph $H$ belonging to
$\Cc$. This proves that \textsc{Distance-$(6r-1)$ Independent Set}
on~$\Cc$ is $\textsc{W}[1]$-hard.
\end{proof}

\bibliographystyle{abbrv}
\bibliography{ref}


\end{document}